\documentclass{amsart}

\usepackage{amssymb,stmaryrd,mathrsfs,dsfont}

\theoremstyle{plain}
\newtheorem{theorem}{Theorem}[section]
\newtheorem{lemma}[theorem]{Lemma}
\newtheorem{proposition}[theorem]{Proposition}
\newtheorem{corollary}[theorem]{Corollary}

\theoremstyle{definition}
\newtheorem{notation}[theorem]{Notation}
\newtheorem{example}[theorem]{Example}
\newtheorem{definition}[theorem]{Definition}

\theoremstyle{remark}
\newtheorem{remark}[theorem]{Remark}

\input xy
\xyoption{all}

\def\au{\mathcal{A}}

\def\auf{(Q, q_0, q_0, \mathcal{F})}

\begin{document}
%%-----------------------------
%%      the top matter
%%-----------------------------
\title[Syntactic Complexity of CSFA]{Syntactic Complexity of Circular Semi-Flower Automata}

\author[S. N. Singh]{Shubh Narayan Singh}
\address{Department of Mathematics, Central University of Bihar, Patna, India}
\email{shubh@cub.ac.in}
\author[K. V. Krishna]{K. V. Krishna}
\address{Department of Mathematics, IIT Guwahati, Guwahati, India}
\email{kvk@iitg.ac.in}

%\date{...}

\begin{abstract}

We investigate the syntactic complexity of certain types of finitely generated submonoids of a free monoid. In fact, we consider those submonoids which are accepted by circular semi-flower automata (CSFA). Here, we show that the syntactic complexity of CSFA with at most one `branch point going in' (bpi) is linear. Further, we prove that the syntactic complexity of $n$-state CSFA with two bpis over a binary alphabet is $2n(n+1)$.

\end{abstract}



\maketitle

\section*{Introduction}
The syntactic complexity of a recognizable language is the cardinality of its syntactic monoid. Further, the syntactic complexity of a class of recognizable languages is the maximal syntactic complexity of languages in that class, taken as a function of the state complexity of these languages. The syntactic complexity of a class of automata is considered to be the syntactic complexity of the class of languages accepted by the automata.  The syntactic complexity of recognizable languages has received more attention in recent years.

In \cite{maslov70}, Maslov observed that $n^n$ is a tight upper bound on the size of the monoid of $n$-state complete and deterministic automata.
Holzer and K\"{o}nig studied the syntactic complexity of unary and binary recognizable languages \cite{holzer04}. For instance, they showed that the syntactic complexity of unary recognizable languages is linear. Also, they proved that if the size of alphabet is at least three, then the syntactic complexity is reached to the maximal size $n^n$. It turns out that the most crucial case is to determine the syntactic complexity of recognizable languages over a binary alphabet. In the binary alphabet case, Holzer and K\"{o}nig  have investigated on the maximal size among all monoids generated by two transformations, where one is a permutation with a single cycle and the other is a non-bijective transformation.

Brzozowski \emph{et al.} investigated the syntactic complexity of various classes of recognizable languages (e.g. \cite{brzowski12tcs,brzowski11,brzowski12dcfs,brzowski12}). Beaudry and Holzer studied the syntactic complexity of reversible deterministic automata \cite{beaudry11}. The syntactic complexity is also studied in \cite{rigo12, rigo11}.

In this work, we restrict the work of Holzer and K\"{o}nig  in \cite{holzer04} to the case of monoids generated by two transformations in which one is a circular permutation and the other is a special type of non-bijective transformation. In particular, we focus on the syntactic complexity of a class of submonoids generated by finite prefix sets of words over a binary alphabet. In this connection, we consider those submonoids which are accepted by circular semi-flower automata.

Semi-flower automata (SFA) have been introduced to study the finitely generated submonoids of a free monoid \cite{giam07,shubh12}. Using SFA, the rank and intersection problem of certain submonids of a free monoid have been investigated \cite{giam08,singh11,singh12}. The circular automata have been studied in various contexts. For instance, the \v{C}ern\'{y} conjecture has been verified for circular automata  \cite{dubuc98, pin78}. Recently, Singh and Krishna have studied the holonomy decomposition of circular SFA \cite{singh13}.

In this paper, we consider circular SFA classified by their bpi(s) -- branch point(s) going in -- and obtain the syntactic complexity of circular SFA. Other than this introduction, the paper has been organized into five sections. In Section 1, we present some preliminary concepts and results that are used in this work. We obtain some necessary properties of circular SFA in Section 2. Our investigations on the syntactic complexity of circular SFA have been presented in sections 3 and 4. Finally, Section 5 concludes the paper.

\section{Preliminaries}

In this section, we provide the necessary background material which shall be useful in this work from \cite{berstel85,lawson04,shubh12}.

We fix our notation regarding functions. We write the argument of a function $\alpha : P \longrightarrow P$ on its left so that $p\alpha$ is the value of the function $\alpha$ at the argument $p$. The composition of functions is designated by concatenation, with the leftmost function understood to apply first so that $p(\alpha \beta) = (p\alpha)\beta$. The function $\alpha$ is said to be \emph{idempotent} if $\alpha^2 = \alpha$. The \emph{rank} of $\alpha$, denoted by $\mbox{rank}(\alpha)$, is the cardinality of the image set $P\alpha$.

Let $A$ be a finite set called an \emph{alphabet} with its elements as \emph{letters}. The free monoid over $A$ is denoted by $A^*$ whose elements are called words, and $\varepsilon$ denotes the empty word -- the identity element of $A^*$. A \emph{language} over $A$ is a subset of $A^*$.

An \emph{automaton} $\au$ over an alphabet $A$ is a quadruple $\au = (Q, I, T, \mathcal{F})$, where $Q$ is a finite set called the set of \emph{states}, $I$ and $T$ are subsets of $Q$ called the sets of \emph{initial} and \emph{final} states, respectively, and $\mathcal{F}\subseteq Q\times A\times Q$ called the set of \emph{transitions}. Clearly, by denoting the states as vertices/nodes and the transitions as labeled arcs, an automaton can be represented by a digraph in which initial and final states shall be distinguished appropriately.

A \emph{path} in $\au$ is a finite sequence $(p_0, a_1, p_1), (p_1, a_2, p_2), \ldots, (p_{k-1}, a_k, p_k)$ of consecutive arcs in its digraph. The word $a_1\cdots a_k \in A^*$ is the \emph{label} of the path. A \emph{null path} is a path from a state to itself labeled by $\varepsilon$. A path that starts and ends at the same state is called as a \emph{cycle}, if it is not a null path. The \emph{language accepted by $\au$}, denoted by $L(\au)$, is the set of words in $A^*$ that are the labels of the paths from an initial state to a final state. A language is \emph{recognizable} if it is accepted by an automaton.

A state $q$ of $\au$ is called a \emph{branch point going in}, in short \emph{bpi}, if the number of transitions coming into $q$ (i.e. the indegree of $q$ -- the number of arcs coming into $q$ -- in the digraph of $\au$) is at least two. We write $BPI(\au)$ to denote the set of all bpis of $\au$.
A state $q$ of $\au$ is \emph{accessible} (respectively, \emph{coaccessible}) if there is a path from an initial state to $q$ (respectively, a path from $q$ to a final state). An automaton is said to be \emph{trim} if all the states of the automaton are accessible and coaccessible. An automaton is said to be \emph{deterministic} if it has a unique initial state and there is at most one transition defined for a state and an input letter. If there is at least one transition defined for a state and an input letter in an automaton, then we say that the automaton is \emph{complete}.

An automaton is called a \emph{semi-flower automaton} (in short, SFA) if it is a trim automaton with a unique initial state that is equal to a unique final state such that all the cycles visit the unique initial-final state. If an automaton $\au = (Q, I, T, \mathcal{F})$ is an SFA, we denote the initial-final state by $q_0$. In which case, we simply write $\au = (Q, q_0, q_0, \mathcal{F})$. An SFA accepts a finitely generated submonoid of the free monoid over the underlying alphabet, and vice versa. Moreover, if an SFA is deterministic, it accepts the submonoid generated by a finite prefix set.

Let $\au = (Q, q_0, T, \mathcal{F})$ be a complete and deterministic automaton over $A$. As there is a unique transition defined over a state and a letter in $\au$, each $a \in A$ induces a function \[\overline{a}: Q \longrightarrow Q\] defined by $q\overline{a} = p$, where $(q, a, p) \in \mathcal{F}$. This phenomenon can be naturally extended to the words in $A^*$. For $x \in A^*$, the function induced by $x$, written $\overline{x} : Q \longrightarrow Q$ is defined inductively as follows. For $q \in Q$, we define $q \overline{\varepsilon} = q$ and, for $u \in A^*$ and $a \in A$, $q(\overline{au}) = (q\overline{a})\overline{u}$.
The set of functions $M(\au) = \{\overline{x} \mid x \in A^*\}$ forms a monoid under the composition of functions, called the \emph{monoid} of $\au$. Note that $M(\au)$ is finite and generated by the functions induced by the letters of $A$. Now, we recall that the automaton $\au$ is \emph{minimal} if and only if $\au$ is accessible and the equivalence relation $\sim_{\au}$ on the state set $Q$ defined by \[p \sim_{\au} q \;\;\mbox{ if and only if }\; \forall x \in A^* \;(p \overline x \in T \Longleftrightarrow q \overline x \in T)\] is the diagonal relation. A minimal automaton for a recognizable language is unique up to isomorphism. The number of states in the minimal automaton of a recognizable language $L$ is called the \emph{state complexity} of $L$.

Let $L$ be a language over an alphabet $A$. The \emph{syntactic congruence} of $L$ is the congruence $\sim_L$ over $A^*$ defined  by
\[u \sim_L v \;\; \mbox{ if and only if }\; \forall x,y \in A^*(xuy \in L \Longleftrightarrow xvy \in L).\]
The quotient monoid $A^*\!/_{\sim_L}$ is called the \emph{syntactic monoid} of $L$. It is well known that $L$ is recognizable if and only if the syntactic monoid of $L$ is finite. Let $L$ be a recognizable language. The \emph{syntactic complexity} of $L$ is defined as the size of the syntactic monoid of $L$. Further, the \emph{syntactic complexity} of a class of recognizable languages is the maximal syntactic complexity of languages in that class, taken as a function of the state complexity of these languages. It is also known that the syntactic monoid of $L$ is isomorphic to the monoid of its minimal automaton. Thus, in order to compute the syntactic complexity of $L$, it is convenient to consider the monoid of its minimal automaton.

We now present the notion of group actions and its related concepts which are useful in this work. For more details, one may refer to any book on basic abstract algebra (e.g.  \cite{foote04}). Let $(H, \circ)$ be a group with identity $e$ and $X$ a nonempty set. A \emph{group action} of $H$ on $X$ is a function $\cdot : X \times H \longrightarrow X$ satisfying the following axioms. For $x \in X$ and $h, h' \in H$, \[x \cdot e = x \;\;\mbox{and}\;\; x\cdot (h\circ h') = (x\cdot h)\cdot h'.\]
For $x \in X$, the \emph{orbit of $x$}, denoted by $\mathcal{O}(x)$,   is the equivalence class of $x$ with respect to the equivalence relation $\sim$ on $X$ defined by \[x \sim y \;\Longleftrightarrow\;  x\cdot h = y \;\;\mbox{for some}\;\; h\in H.\] Clearly, $\mathcal{O}(x) = \{x\cdot h\;|\;h \in H\}$. Further, for $x \in X$, the \emph{stabilizer of $x$}, denoted by $H_x$, is the subgroup of $H$ defined by \[H_x = \{h \in H\;|\; x\cdot h = x\}.\]
It can be observed that, for $x \in X$, $|\mathcal{O}(x)| = [H, H_x]$, the index of $H_x$ in $H$.

\section{Circular Semi-Flower Automata}

In this section, we recall some necessary properties of circular semi-flower automata (CSFA) from \cite{singh13} and prove further properties which are useful in determining the syntactic complexity of CSFA. In this paper, all the automata are complete and deterministic.

Let $X = \{p_1, \ldots, p_m\}$ be a nonempty finite set. A function $\alpha$ on $X$ is said to be a \emph{circular permutation} on $X$ if there is a cyclic ordering, say $p_{i_1}, \ldots, p_{i_m}$, on the elements of $X$, i.e. \[p_{i_j}\alpha = p_{i_{j+1}}, \mbox{ for $1 \le j < m$, and }\; p_{i_m}\alpha = p_{i_1}.\]
An automaton $\au$ over $A$ is said to be a \emph{circular automaton} if there exists $a \in A$ such that the function $\overline{a}$ is a circular permutation on the state set of $\au$. Now we recall the following result.

\begin{theorem} [\cite{singh13}] \label{c3.l.ucp}
Let $\au$ be an SFA over $A$.
\begin{enumerate}
\item[(i)] For $a \in A$, if $\overline{a}$ is a permutation on $Q$, then $\overline{a}$ is a circular permutation.
\item[(ii)] For $a, b \in A$, if $\overline{a}$ and $\overline{b}$ are permutations on $Q$, then $\overline{a} = \overline{b}$.
\item[(iii)] $BPI(\au) = \varnothing$ if and only if $|A| = 1$.
\end{enumerate}
\end{theorem}

Unless otherwise stated, in what follows, $\au$ always denotes a CSFA $\auf$ over $A$ such that $|Q| = n$. In view of Theorem \ref{c3.l.ucp}, $\au$ has a unique circular permutation on $Q$, induced by its input letters. For the rest of the paper, we fix the following regarding $\au$. Assume $a \in A$ induces a circular permutation $\overline{a}$ and accordingly \[q_0, q_1, \ldots, q_{n-1}\] is the cyclic ordering on $Q$ with respect to $\overline{a}$. And let $G$ be the submonoid generated by $\overline a$ in the monoid $M(\au)$.

\begin{proposition}
The submonoid $G$ is a cyclic subgroup of order $n$ in $M(\au)$. Further, $G$ contains all the permutations of $M(\au)$.
\end{proposition}

\begin{proof}
It is straightforward to observe that $G$ is a cyclic group of order $n$, because $G$ is the submonoid generated by the circular permutation $\overline{a}$. Now, let $\overline x \in M(\au)$ be a permutation on $Q$ for $x = a_1 a_2 \cdots a_m$ with $a_i \in A$ $(1\le i \le m)$. Then \[\overline x = \overline{a_1 a_2 \cdots a_m} =  \overline a_1  \overline a_2 \cdots \overline a_m.\] Clearly, each function $\overline {a_i}$ is a permutation on $Q$. By Theorem \ref{c3.l.ucp}(ii), we have $\overline a = \overline {a_i}$, for all $i$ $(1\le i \le m)$. This implies that $\overline x = \overline{a^m}$ and consequently $\overline x \in G$.
\end{proof}

\begin{remark}\label{c4.r.paiq}
For $p,q \in  Q$, there exists $\overline x \in G$ such that $p\overline x = q$. Indeed, if $p = q_i$ and $q = q_j$, for some $i$,  $j$ (with $0 \le i \le j < n$), then $\overline x = \overline{a^{j-i}}$ will serve the purpose.
\end{remark}

\begin{proposition}\label{c4.p.csm}
$\au$ is a minimal automaton.
\end{proposition}

\begin{proof}
Since $\au$ is accessible, it is sufficient to prove that the relation $\sim_{\au}$ is diagonal. Let $p$ and $q$ be two distinct states. By Remark \ref{c4.r.paiq}, there exists $\overline x \in G$ such that $p \overline{x} = q_0$.
Now, we claim that $q \overline{x} \neq q_0$. For, if $q \overline{x} = q_0$, then $p \overline{x} = q\overline{x}$. Since $\overline x \in G$, we have $p = q$; a contradiction. Hence, the relation $\sim_{\au}$ is diagonal and consequently, $\au$ is minimal.
\end{proof}

Let us consider the group action of $G$ on $M(\au)$ with respect to the monoid operation, the composition of functions. Note that \[M(\au) = \displaystyle\bigcup_{ x \in A^*}\mathcal{O}(\overline x).\]

\begin{proposition}\label{c4.p.sizen}
For $x \in A^*$, we have $|\mathcal{O}(\overline x)| = n$.
\end{proposition}

\begin{proof}
For $x \in A^*$, we have $|\mathcal{O}(\overline x)| = [G, G_{\overline x}]$. Since $|G| = n$, it is sufficient to prove that $G_{\overline x} = \{\overline {\varepsilon}\}$. Let $\overline y \in G_{\overline x}$, we have $\overline x \;\overline y = \overline x$. This implies that, for $q \in Q$, $q(\overline x\; \overline y) = q\overline x$, i.e. $(q\overline x) \overline y = q\overline x$. Write $q\overline x =q'$, then $q' \overline y = q'$.

We claim that $\overline y = \overline{\varepsilon}$. Let $p \in Q$ be an arbitrary state. By Remark \ref{c4.r.paiq}, there exists $\overline z \in G$ such that $p = q'\overline z$. Consider
\[p\overline y = (q'\overline z)\overline y = (q'\overline y)\overline z = q'\overline z = p.\]
Hence, $\overline y = \overline{\varepsilon}$ and consequently $G_{\overline x} = \{\overline {\varepsilon}\}$.
\end{proof}

Thus, to compute the syntactic complexity, it is sufficient to count the number of orbits. In the rest of the paper, we investigate the syntactic complexity of CSFA classified by the number of bpis. The following result from \cite{singh13} is useful in the sequel.

\begin{theorem}\label{t.pre-pap}
For $k \ge 1$, if $|BPI(\au)| = k$, then
\begin{enumerate}
\item[(i)] $q_0 \in BPI(\au)$ and
\item[(ii)] any non-permutation in $M(\au)$ has rank at most $k$.
\end{enumerate}
Hence, if $k = 1$, then  $Q\overline b = \{q_0\}$, for all $b \in A\setminus \{a\}$.
\end{theorem}

\section{CSFA with at most one bpi}

In this section, we investigate the syntactic complexity of CSFA with at most one bpi. We first observe that the syntactic complexity of SFA with no bpis follows from the general case of permutation SFA. An automaton is a \emph{permutation automaton} if the function induced by each
input letter is a permutation on the state set \cite{thi68}. By Theorem \ref{c3.l.ucp}(i) and Proposition \ref{c4.p.csm}, any permutation SFA is a minimal automaton. Now, by Theorem \ref{c3.l.ucp}(ii), we have the following proposition which also provides the syntactic complexity of permutation SFA.

\begin{proposition}\label{c4.p.scpa}
If $\au$ is a permutation SFA, then $M(\au)$ is a cyclic group of order $n$. Hence, the syntactic complexity of $\au$ is $n$.
\end{proposition}

Let $\au$ be an SFA with no bpis, then by Theorem \ref{c3.l.ucp}(iii), we have $|A| = 1$, say $A = \{a\}$. Note that the function $\overline a$ is a circular permutation on $Q$. Thus, $\au$ is a circular as well as permutation SFA.  Hence, by Proposition \ref{c4.p.scpa}, we have the following theorem.

\begin{theorem}
The syntactic complexity of SFA with no bpis is $n$.
\end{theorem}

Now, we determine the syntactic complexity of CSFA with a unique bpi. If the size of state set $|Q| = 1$, then the CSFA with a unique bpi is a permutation SFA so that its syntactic complexity is $n = 1$. Now, we consider the CSFA with $|Q| > 1$ in the following theorem.

\begin{theorem}
The syntactic complexity of CSFA with a unique bpi is $2n$.
\end{theorem}

\begin{proof}
Let $\au$ be a CSFA with a unique bpi. By Theorem \ref{t.pre-pap}, we have $Q\overline b = \{q_0\}$, for all $b \in A\setminus \{a\}$. This implies that for $b, c \in A\setminus \{a\}$, we have $\overline b = \overline c$. Now, we take a letter $b \in A\setminus \{a\}$. The orbit of $\overline b$ is \[\mathcal{O}(\overline b) = \{\overline{ba^i}\;|\; 1 \leq i \leq n\}.\]

Let $\overline x$ be a non-permutation in $M(\au)$. By Theorem \ref{t.pre-pap}(ii), $\overline x$ is a constant function. This implies that $Q \overline x = \{q_k\}$, for some $k$ (with $0\le k < n$). Note that $Q \overline{ba^k} = \{q_k\}$. Therefore, $\overline x = \overline{ba^k} \in \mathcal{O}(\overline b)$ and consequently the orbit $\mathcal{O}(\overline b)$ contains all non-permutations in $M(\au)$.

Thus, there are exactly two distinct orbits, one with all permutations (i.e. $G$) and other with all non-permutations. By Proposition \ref{c4.p.sizen},
we have $|M(\au)| = 2n$. Since $\au$ is arbitrary, we have the syntactic complexity of the submonoids accepted by  CSFA with a unique bpi is $2n$.
\end{proof}

\section{CSFA with two bpis}

In this section, we investigate the syntactic complexity of CSFA with two bpis. In the previous section, we have observed that the syntactic complexity of CSFA with at most one bpi is independent of the size of the input alphabet. In contrast, the syntactic complexity of CSFA with two bpis varies with respect to the size of input alphabet (cf. Example \ref{c4.e.2and3}). Hence, in this section, we restrict ourselves to investigate the syntactic complexity of CSFA with two bpis over a binary alphabet. First observe that, if $|Q| = 2$, then the CSFA under consideration are indeed permutation SFA so that their syntactic complexity is $n = 2$. In this section, we consider the CSFA with $|Q| > 2$ and prove the following main theorem.

\begin{theorem}\label{c4.t.mainres}
The syntactic complexity of CSFA with two bpis over a binary alphabet is $2n(n+1)$.
\end{theorem}

We fix the following notation for rest of the section. Let $\au$ be a CSFA with two bpis over the binary alphabet $A = \{a,b\}$. As earlier, $\overline a$ is the circular permutation. Note that, for the non-permutation $\overline b$, we have $Q\overline b = BPI(\au)$. By Theorem \ref{t.pre-pap}(i), the initial-final state $q_0$ is a bpi. Let $q_m$, for some $m$ (with $1 \le m < n$), be the other bpi of $\au$ so that $BPI(\au) = \{q_0, q_m\}$. We need to establish some results for proving Theorem \ref{c4.t.mainres}. In the following, these results are presented in various subsections.

\subsection{Idempotents}

In this subsection, we obtain the idempotents of $M(\au)$ which will be useful to give a representation of the elements of $M(\au)$. In view of Theorem \ref{t.pre-pap}(ii), for $x \in A^*$, we have $\mbox{rank}(\overline x) \in \{1, 2, n\}$.
Clearly, the identity element $\overline{\varepsilon}$ in $M(\au)$ is only idempotent of rank $n$. All the elements of rank one in $M(\au)$ are idempotent, provided that they exist. We now estimate idempotents of rank two in $M(\au)$. For that, we first prove the following results.

\begin{remark}\label{c4.p.ndie}
For $1 \leq i \leq n$ and $x \in A^*$, if $\overline x$ is an idempotent in $M(\au)$, then
$\overline{a^i x a^{n-i}}$ is also an idempotent in $M(\au)$. For instance, \[(\overline{a^i x a^{n-i}})^2 = (\overline{a^i x a^{n-i}}) (\overline{a^i x a^{n-i}}) = \overline{a^i x^2 a^{n-i}} = \overline{a^i x a^{n-i}}.\]
\end{remark}

\begin{remark}\label{c4.r.q0qm-idem}
Let $\overline x$ be an element in $M(\au)$ such that $Q\overline x = \{q_0, q_m\}$. If $q_0\overline x = q_0$ and $q_m \overline x = q_m$, then $\overline x$ is an idempotent.
\end{remark}

\begin{proposition}\label{c4.p.number}
Let $t$ be a natural number such that $t < m < n$; there exists a natural number $k$ such that $m \le t + k(n-m) < n$.
\end{proposition}

\begin{proof}
Since $n-m > 0$, note that the sequence $\{t + i(n-m)\}_{i = 0,1,2,\ldots}$ is an increasing sequence. Let $k$ be the least number such that $m\le t + k(n-m)$. We prove that $t + k(n-m)< n$. Since $k$ is least, we have $t + (k-1)(n-m)< m$. This implies that \[t + (k-1)n - km < 0.\]
Now, we have $t + k(n-m) = t + (k-1)n -km + n < n$.
\end{proof}

\begin{lemma}\label{c4.l.akb}
There exists a natural number $r$ (with $1 \leq r < n$) such that the function $\overline{a^r b}$ is an idempotent of rank two in $M(\au)$.
\end{lemma}

\begin{proof}
Since $q_m$ is the bpi of $\au$, there exists $j$ (with $0 \le j < m$) such that $q_j \overline b = q_m$. Let $t$ (with $0\le t < m$) be the least number such that $q_t \overline b = q_m$ so that $q_0 \overline {a^t b} = q_m$. Consequently, as $q_m \overline{a^{n-m}} = q_0$, we have \[q_m \overline{a^{n-m+t}b} = q_m.\]
If $q_0 \overline{a^{n-m+t}b} = q_0$, then choose $r = n-m+t$ and by Remark \ref{c4.r.q0qm-idem}, the function $\overline{a^r b}$ is an idempotent of rank two in $M(\au)$. Otherwise, since the letter $b$ is suffix of word $a^{n-m+t}b$, we have $q_0 \overline{a^{n-m+t}b} = q_m$. Then \[q_m \overline{a^{2(n-m)+t}b} = q_m.\] If $q_0 \overline{a^{2(n-m)+t}b} = q_0$, then choose $r = 2(n-m)+t$ and again by Remark \ref{c4.r.q0qm-idem}, the function $\overline{a^r b}$ is an idempotent of rank two in $M(\au)$. Otherwise, since the letter $b$ is suffix of word $a^{2(n-m)+t}b$, we have $q_0 \overline{a^{2(n-m)+t}b} = q_m$. Then \[q_m  \overline{a^{3(n-m)+t}b}  = q_m.\]
As long as we continue this process, in each $i^{th}$ step, we have $q_m  \overline{a^{i(n-m)+t}b}  = q_m.$
Note that, by  Proposition \ref{c4.p.number}, there exists a natural number $k$ such that \break$m \leq k(n-m)+t < n$. If the above process terminates with a number $r$ before $k^{th}$ step, then we are through.  Otherwise, in the $k^{th}$ step, we have $q_m  \overline{a^{k(n-m)+t}b}  = q_m.$ Moreover, since $m \leq k(n-m)+t < n$, \[q_0 \overline {a^{t +k(n-m)}b} = q_{k(n-m)+t}\overline{b} = q_0.\] Thus, choose $r = k(n-m) + t$, and hence by Remark \ref{c4.r.q0qm-idem}, the function $\overline{a^r b}$ is an idempotent of rank two in $M(\au)$.
\end{proof}

\begin{notation}
The number $k(n-m) + t$ obtained in Lemma \ref{c4.l.akb} is always denoted by $\kappa$ so that $\overline {a^\kappa b}$ is an idempotent of rank two in $M(\au)$.
\end{notation}

\begin{lemma}\label{c4.l.1b1}\
\begin{enumerate}
\item[\rm(i)] If $q_0\overline b \neq q_0$, then $\overline b^2$ is an idempotent of rank two in $M(\au)$.
\item[\rm(ii)] If $q_0\overline b = q_0$, then there exists $t$ (with $1 \leq t < m$) such that the function $(\overline {a^t b})^2$ is an idempotent of rank two in $M(\au)$.
\end{enumerate}
\end{lemma}

\begin{proof}
We know that $q_m\overline b = q_0$.

(i) Since $q_0\overline b \neq q_0$, we have $q_0 \overline b = q_m$. Consider $Q\overline b^2 = (Q\overline b)\overline b = \{q_0, q_m\}\overline b = \{q_0, q_m\}$. Also, since $q_0 \overline b^2 = q_0$ and $q_m \overline b^2 = q_m$, by Remark \ref{c4.r.q0qm-idem}, the function $\overline b^2$ is an idempotent of rank two in $M(\au)$.

(ii) Since $q_0\overline b = q_0$, the state $q_1$ is not a bpi. Therefore, $1 < m < n$. Further, there exists $j$ (with $0 < j < m$) such that $q_j \overline b = q_m$. Let $t$ (with $1\le t < m$) be the least number such that $q_t \overline b = q_m$ so that $q_0 \overline {a^t b} = q_m$. We claim that $q_m\overline {a^t b} = q_0$.

On the contrary, assume that $q_m\overline {a^t b} \ne q_0$. Then, $q_m\overline {a^t b} = q_m$ so that there is a cycle from $q_m$ to $q_m$ labeled by $a^tb$. Since $\au$ is an SFA, the cycle should pass through $q_0$.  Since $q_0\overline b = q_0$, there exist $t_1$ and $t_2$ ($1 \leq t_1, t_2 < t$) with $t_1 + t_2 = t$ such that\[q_m \overline {a^{t_1}} = q_0 \;\;\mbox{and}\;\; q_0 \overline{a^{t_2}b} = q_m.\] Note that $q_0 \overline{a^{t_2}b} = q_{t_2}\overline{b} = q_m$. This contradicts the choice of $t$, as $t_2 < t$. Thus, $q_m\overline {a^t b} = q_0$.

Now, observe that $Q (\overline {a^t b})^2 = (Q \overline {a^t b})\overline {a^t b}  = \{q_0, q_m\} \overline{a^t b} = \{q_0, q_m\}$. Further, $q_0 (\overline {a^t b})^2= q_0$ and $q_m (\overline {a^t b})^2 = q_m$. By Remark \ref{c4.r.q0qm-idem}, the function $(\overline{a^t b})^2$ is an idempotent of rank two in $M(\au)$.
\end{proof}

\begin{notation}
In this section, $\tau$ always denotes the number obtained in Lemma \ref{c4.l.1b1}(ii). That is, if $q_0\overline b = q_0$, $\tau$ is the least number such that $q_\tau\overline b = q_m$, so that $(\overline {a^\tau b})^2$ is an idempotent of rank two in $M(\au)$.
\end{notation}

In view of Remark \ref{c4.p.ndie}, we have the following corollary of Lemma \ref{c4.l.akb} and Lemma \ref{c4.l.1b1}.

\begin{corollary}For $1\le i \le n$,
\begin{enumerate}
\item[\rm(i)] $\overline {a^i (a^\kappa b) a^{n-i}}$ is an idempotent of rank two in $M(\au)$.
\item[\rm(ii)]  If $q_0\overline b \neq q_0$, then $\overline {a^i b^2 a^{n-i}}$ is an idempotent of rank two in $M(\au)$.
\item[\rm(iii)] If $q_0\overline b = q_0$, then $\overline {a^i(a^\tau b)^2 a^{n-i}}$ is an idempotent of rank two in $M(\au)$.
\end{enumerate}
\end{corollary}

\begin{definition}
We call the following list of $2n + 2$ idempotents, if they exist, in $M(\au)$ as the \emph{basic idempotents}. The set of all the basic idempotents in $M(\au)$ is denoted by $\mathfrak{B}$.
\begin{enumerate}
\item[\rm(i)] The idempotent $\overline{\varepsilon}$.
\item[\rm(ii)] The idempotent whose image set is $\{q_0\}$, denoted by $\overline{\nu}$.
\item[\rm(iii)] For $1\leq i\leq n$, the idempotent $\overline{a^i (a^\kappa b) a^{n-i}}$.
\item[\rm(iv)] For $1\leq i\leq n$, if $q_0 \overline b \neq q_0$, then the idempotent $\overline{a^i b^2 a^{n-i}}$;
 else, the idempotent $\overline{a^i (a^\tau b)^2 a^{n-i}}$.
\end{enumerate}
\end{definition}

\begin{remark}\label{c4.r.2n1}
Clearly, $|\mathfrak{B}| \leq 2(n + 1)$.
\end{remark}

The following example shows that the cardinality of the set of basic idempotents is not necessarily $2(n+1)$.

\begin{example}
Note that the automaton given in \textsc{Figure} \ref{c4.f.basic-idem} is a CSFA in which $q_0$ and  $q_2$ are the bpis. For this CSFA, it can be observed that $\overline b^2 = \overline{a^\kappa b}$, where $\kappa = 2$. Hence, $|\mathfrak{B}| < 2(n + 1)$

\begin{figure}[htp]
\entrymodifiers={++[o][F-]} \SelectTips{cm}{}
\[\xymatrix{*\txt{} & *++[o][F=]{q_0} \ar[dr]^a \ar@/^0.7pc/[dd]^b *\txt{}\\
q_3 \ar[ur]^a \ar@/^1.5pc/[ur]^b & *\txt{} & q_1\ar[dl]^a \ar@/^1.5pc/[dl]^b\\
*\txt{} & q_2 \ar[ul]^a \ar@/^0.7pc/[uu]^b & *\txt{}}\]
\caption{A CSFA with two bpis}
\label{c4.f.basic-idem}
\end{figure}
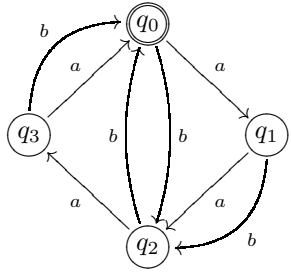
\end{example}

\subsection{Elements of rank two}

In this subsection, we obtain a representation of the elements of rank two in $M(\au)$. Here, we recall the definition of the complement of a function of rank two from \cite{kraw05}.

\begin{definition}
Let $X$ be a nonempty finite set and $\alpha$ a function on $X$ such that $X\alpha = \{i,j\}$. The \emph{complement} of $\alpha$ is the function $\alpha^\#$ defined by, for $k \in X$,
\begin{equation*}
k \alpha^\# =
\begin{cases}
i & \text{if $k \alpha = j$};\\
j & \text{if $k \alpha = i$.}
\end{cases}
\end{equation*}
\end{definition}

The following lemma is useful in the sequel.

\begin{lemma}\label{c4.l.conj-elem}\
\begin{enumerate}
\item[\rm(i)] If $q_0\overline b \neq q_0$, then $\overline b^\# = \overline b^2$.
\item[\rm(ii)]If $q_0\overline b = q_0$, then $\overline b^\# =\overline{ba^\tau b}$.
\end{enumerate}
\end{lemma}

\begin{proof}
We recall that $q_m\overline b = q_0$ and $Q\overline b = \{q_0, q_m\}$. Note that, for $q \in Q$, either $q \overline b = q_0$ or $q \overline b = q_m$.
\begin{itemize}
\item[(i)] Since $q_0\overline b \neq q_0$, we have $q_0\overline b = q_m$. Let $q \in Q$. If $q \overline b = q_0$, then \[q \overline b^2 = (q\overline b)\overline b = q_0 \overline b = q_m.\]
Else, \[q \overline b^2 = (q\overline b)\overline b = q_m \overline b = q_0.\] Hence, $\overline b^\# = \overline b^2$.

\item[(ii)] Given $q_0\overline b = q_0$. Let $q \in Q$. If $q \overline b = q_0$, then  \[q \overline{ba^\tau b} = (q\overline b)\overline{a^\tau b} = q_0\overline{a^\tau b} = q_m\] (cf. Lemma \ref{c4.l.1b1}(ii)). Else, \[q \overline{ba^\tau b} = (q\overline b)\overline{a^\tau b} = q_m \overline{a^\tau b} = q_0.\] Hence, $\overline b^\# = \overline{b a^\tau b}$.
\end{itemize}
\end{proof}

\begin{theorem}\label{c4.t.rank2form}
Any element of rank two in $M(\au)$ has one of the following forms.
\begin{enumerate}
\item[$(\beta)$] $\overline{a^i b a^j}$
\item[$(\gamma)$] $\overline{a^i b^2 a^j}$
\item[$(\delta)$] $\overline{a^i b a^\tau b a^j}$
\end{enumerate}
Here, $i,j\in \{1,\ldots, n\}$.
\end{theorem}

\begin{proof}
Note that every element of rank two in $M(\au)$ should have at least one $b$. Let $w = a^{i_1} b a^{i_2} b \ldots ba^{i_{k-1}} b a^{i_k} \in A^*$, for $i_t \geq 0$ ($t \in \{1,\ldots,k\}$), such that $\overline w$ be an arbitrary element of rank two in $M(\au)$.
Write $w = a^{i_1} b u b a^{i_k}$, where $u = a^{i_2} b \ldots ba^{i_{k-1}}$. Clearly, the function $\overline{b u b}$ has rank two with the image set $\{q_0, q_m\}$.

\begin{description}
  \item[Case-1 ($\overline{b u b} = \overline b$)] Clearly, $\overline w = \overline{a^{i_1}b u b a^{i_k}} = \overline{a^{i_1} b a^{i_k}}$, which is in the form $(\beta)$.
  \item[Case-2 ($\overline{b u b} \ne \overline b$)] First we claim that $\overline{b u b} = \overline b^\#$. Since $\overline{b u b} \ne \overline b$, there exist $p \in Q$ such that $p\overline{b u b} \ne p\overline b$. Now, we consider two subcases according to the state $p\overline{b}$.
\begin{description}
  \item[Subcase-1 ($p\overline{b} = q_0$)] Since $\overline{b u b} \ne \overline b$, we have $p\overline{b u b} = q_m$. Consequently, $$q_0\overline{u b} = q_m.$$ Let $q \in Q$ be an arbitrary element. Then, either $q\overline b = q_0$ or $q\overline b = q_m$.

If $q\overline b = q_0$, then \[q\overline{b u b} = (q\overline b) \overline{u b} = q_0\overline{u b} = q_m.\] Else ($q\overline b = q_m$),
$q\overline{b u b} = (q\overline b) \overline{u b} = q_m\overline{u b}.$ To show the last term is equal to $q_0$, let us assume the contrary. That is, assume $q_m\overline{u b} \ne q_0$. Then, $q_m\overline{u b} = q_m$. Consequently, \[Q\overline{bub} = (Q\overline b)\overline{ub} = \{q_0, q_m\}\overline{ub} = \{q_m\}.\] This is a contradiction to $\overline{bub}$ is of rank two. Thus, if $q\overline b = q_m$, then $q\overline{b u b} = q_0$. Hence, $\overline{b u b} = \overline{b}^\#$.
  \item[Subcase-2 ($p\overline{b} \ne q_0$)] One can proceed in the similar lines as in Subcase-1 and obtain that $\overline{b u b} = \overline{b}^\#$.\\
\end{description}

\noindent If $q_0\overline b \neq q_0$, then by Lemma \ref{c4.l.conj-elem}(i), we have $\overline b^\# = \overline b^2$. Consequently,
\[\overline w = \overline{a^{i_1}b u b a^{i_k}} = \overline{a^{i_1} b^2 a^{i_k}},\] which is in the form $(\gamma)$.

\noindent If $q_0\overline b = q_0$, then by Lemma \ref{c4.l.conj-elem}(ii), we have $\overline b^\# = \overline{b a^\tau b}$. Consequently,
\[\overline w = \overline{a^{i_1}b u b a^{i_k}} = \overline{a^{i_1} b a^\tau b a^{i_k}},\] which is in the form $(\delta)$.
\end{description}
\end{proof}

\subsection{Representation of $M(\au)$}

In this subsection, we give a canonical representation of the elements of $M(\au)$ in terms of basic idempotents and circular permutation.

\begin{theorem}\label{c4.t.cano-form}
Every element of $M(\au)$ can be written as a composition of a basic idempotent and a permutation, i.e.
\[M(\au) = \mathfrak{B}G = \Big\{\overline{e} \;\overline{g}\;\Big|\;\overline{e} \in \mathfrak{B}\;\mbox{and}\; \overline{g} \in G  \Big\}.\]
\end{theorem}

\begin{proof}
For $x \in A^*$, by Theorem \ref{t.pre-pap}(ii), we have  $\mbox{rank}(\overline x)\in \{1,2,n\}$. If $\mbox{rank}(\overline x) = 1$, then the function $\overline x$ is an idempotent (being a constant function). Therefore, there exists $i$ (with $1 \le i \le n$) such that
\[\overline x = \overline{\nu}  \;\overline{a^i} \in \mathfrak{B}G.\]
If $\mbox{rank}(\overline x) = n$, then the function $\overline x$ is a permutation of the form $\overline x = \overline{a^i}$, for some $i$ (with $1 \le i \le n$). Clearly, $\overline x \in G$ so that \[\overline x = \overline{\varepsilon}\;\overline x  \in \mathfrak{B}G.\]
If $\mbox{rank}(\overline x) = 2$, then, by Theorem \ref{c4.t.rank2form}, $\overline x = \overline{a^i b a^j}$ or $\overline x = \overline{a^i b^2 a^j}$ or $\overline x = \overline{a^i b a^\tau b a^j}$, for some $i, j \in \{1, \ldots, n\}$.\\

\begin{itemize}
\item[If $\overline x  =$] $\overline{a^i b a^j}$,  then
\[\overline x  = \overline{a^{i-\kappa} (a^\kappa b) a^j}= \overline{a^{i-\kappa} (a^\kappa b)a^{n-(i-\kappa)}}\;\; \overline{a^{j+(i-\kappa)}}= \overline{a^{i'} (a^\kappa b)a^{n-i'}} \;\;\overline{a^{j'}},\] where $i'$ and $j'$ are, respectively, the residues of $(i-\kappa)$ and $(j+i-\kappa)\mod n$. Consequently, $\overline x \in \mathfrak{B}G$.\\

\item[If $\overline x  =$] $\overline{a^i b^2 a^j}$, then \[\overline x = \overline{a^{i} b^2 a^{n-i}}\;\;\overline{a^{j-(n-i)}} = \overline{a^{i} b^2 a^{n-i}}\;\;\overline{a^{j'}},\] where $j'$ is the residue of $(j+i-n)\mod n$. Consequently, $\overline x \in \mathfrak{B}G$.\\

\item[If $\overline x  =$] $\overline{a^i b a^\tau b a^j}$, then
\[\overline x =  \overline{a^{i-\tau} (a^\tau b)^2  a^j} = \overline{a^{i-\tau} (a^\tau b)^2 a^{n-(i-\tau)}}\;\; \overline{a^{j-n+(i-\tau)}} = \overline{a^{i'} (a^\tau b)^2 a^{n-i'}}\;\; \overline{a^{j'}},\]
where $i'$ and $j'$ are, respectively, the residues of $(i-\tau)$ and $(j+i-\tau -n)\mod n$. Consequently, $\overline x \in \mathfrak{B}G$.
\end{itemize}
Thus, in all the cases the function $\overline x \in M(\au)$ can be written as a composition of a basic idempotent and a permutation in $G$. Hence,
$M(\au) = \mathfrak{B}G$.
\end{proof}

\subsection{An example}

Consider the CSFA $\au' = (Q, 1,1,\mathcal{F})$ over $A = \{a,b\}$ with $Q = \{1,2,\ldots, n\}$, and the transitions are given in the following table.
\[\begin{array}{c|cccccc}
      \mathcal{F} & 1 & 2 & 3 & \cdots & n-1 & n \\
\hline
      a & 2 & 3 & 4 & \cdots & n & 1 \\
      b & 2 & 1 & 1 & \cdots & 1 & 1 \\
    \end{array}\]
Clearly, the input letters $a$ and $b$ induces the functions $\overline a$ and $\overline b$ on $Q$, respectively, given as
\[\overline a = \left(
  \begin{array}{cccccc}
    1 & 2 & 3 & \cdots & n-1 & n \\
    2 & 3 & 4 & \cdots & n & 1 \\
  \end{array}
\right)\mbox{and}\;
\overline b =
\left(
\begin{array}{cccccc}
1 & 2 & 3 & \cdots & n-1 & n \\
2 & 1 & 1 & \cdots & 1 & 1 \\
\end{array}
\right).\]
One can observe that $Q \overline{bab} = \{1\}$. Therefore, the function $\overline{bab}$ is the constant function $\overline \nu$ in $M(\au')$.  Further, we observe that $\kappa = n-1$ and the functions
\[\overline b^2 =
\left(
\begin{array}{cccccc}
1 & 2 & 3 & \cdots & n-1 & n \\
1 & 2 & 2 & \cdots & 2 & 2 \\
\end{array}
\right)\mbox{and}\;
\overline{a^\kappa b} = \left(
  \begin{array}{cccccc}
    1 & 2 & 3 & \cdots & n-1 & n \\
    1 & 2 & 1 & \cdots & 1 & 1 \\
  \end{array}
\right)\]
are idempotents of rank two in $M(\au')$. By Remark \ref{c4.p.ndie}, the functions $\overline{a^i b^2 a^{n-i}}$ and $\overline{a^i (a^{n-1}b) a^{n-i}}$ are basic idempotents of rank two in $M(\au')$, where $i \in \{1,2,\ldots,n\}$. Now, we pursue on the orbits of basic idempotents of rank two. In this connection, first note that, for $1 \le r \le n$,
\[\overline{a^r b} = \left(
  \begin{array}{ccccccccc}
    1 & 2  & \cdots &n-r & n-r+1 & n-r+2 &\cdots & n-1 & n \\
     1& 1  & \cdots &1 & 2 & 1 & \cdots & 1 & 1 \\
  \end{array}
\right),\]
\[\overline{a^r b^2} = \left(
  \begin{array}{cccccccccc}
    1 & 2 & \cdots &n-r & n-r+1 & n-r+2  &\cdots & n-1 & n \\
     2& 2 & \cdots & 2& 1 & 2 &\cdots & 2 & 2 \\
  \end{array}
\right).\]
For $1\le j < i\le n$, let us assume that $\mathcal{O}(\overline{a^i b^2 a^{n-i}}) \cap \mathcal{O}(\overline{a^j b^2 a^{n-j}}) \ne \varnothing$. Then, for some $t$ (with $1\le t\le n$),
\[\overline{a^i b^2 a^{n-i}} = \overline{a^j b^2 a^{n-j}} \;\;\overline{a^t} \Longrightarrow \overline{a^{i-j} b^2} = \overline{b^2 a^{i-j+t}}.\]
If $i-j+t \ne 0(\bmod \; n)$, then  $Q\overline{b^2 a^{i-j+t}}\ne \{1,2\} = Q\overline{a^{i-j} b^2}$; a contradiction.
Otherwise, we have $\overline{a^{i-j} b^2} = \overline{b^2}$. But, from the above shown $\overline b^2$ and $\overline{a^rb^2}$, we can observe that $\overline{a^{i-j} b^2} \ne \overline{b^2}$. Hence, for $1\le j < i\le n$, we have \[\mathcal{O}(\overline{a^i b^2 a^{n-i}}) \cap \mathcal{O}(\overline{a^j b^2 a^{n-j}}) = \varnothing.\]
Similarly, we can prove that, for $1\le j < i \le n$, we have \[\mathcal{O}(\overline{a^i (a^{n-1}b) a^{n-i}}) \cap \mathcal{O}(\overline{a^j (a^{n-1}b) a^{n-j}}) = \varnothing.\]
Note that $\overline b^2 \ne \overline{a^{n-1}b}$. Now, for $1\le j < i\le n$, let us assume that $$\mathcal{O}(\overline{a^i b^2 a^{n-i}}) \cap \mathcal{O}(\overline{a^j (a^{n-1}b) a^{n-j}}) \ne \varnothing.$$ Then, for some $t$ (with $1\le t\le n$), we have
\[\overline{a^i b^2 a^{n-i}} = \overline{a^j (a^{n-1}b) a^{n-j}} \;\;\overline{a^t} \Longrightarrow \overline{a^{i-j} b^2} = \overline{(a^{n-1}b) a^{i-j+t}}.\]
If $i-j+t \ne 0(\bmod \; n)$, then  $Q\overline{b^2 a^{i-j+t}}\ne \{1,2\} = Q\overline{a^{i-j} b^2}$; a contradiction.
Otherwise, we have $\overline{a^{i-j} b^2} = \overline{ a^{n-1}b}$. But, from the above shown $\overline{a^{n-1}b}$ and $\overline{a^rb^2}$, we can  observe that $\overline{a^{i-j} b^2} \ne \overline{ a^{n-1}b}$. Hence, for $1\le j < i\le n$, we have \[\mathcal{O}(\overline{a^i b^2 a^{n-i}}) \cap \mathcal{O}(\overline{a^j (a^{n-1}b) a^{n-j}}) = \varnothing .\]  Thus, all the orbits of the basic idempotents of rank two are disjoint and so all the basic idempotents of rank two are distinct. Thus, $|\mathfrak{B}| = 2(n+1)$. Consequently, $M(\au') = \mathfrak{B}G = 2n(n+1)$. Hence, the syntactic complexity of the CSFA $\au'$ is $2n(n+1)$.

\subsection{Proof of Theorem \ref{c4.t.mainres}}

Now, we prove the main Theorem \ref{c4.t.mainres}. We know that
\begin{eqnarray*}
M(\au) &=& \bigcup_{\overline x \in M(\au)}\mathcal{O}(\overline x)\\
&=& \bigcup_{\overline x \in \mathfrak{B}G}\mathcal{O}(\overline x)\; \mbox{ by using Theorem \ref{c4.t.cano-form}}\\
&=& \bigcup_{\overline x \in \mathfrak{B}}\mathcal{O}(\overline x).
\end{eqnarray*}

This implies that
\begin{eqnarray*}
|M(\au)| &\le& |\mathfrak{B}||\mathcal{O}(\overline x)|\\
&\leq& 2n(n+1) \; \mbox{ by using Proposition \ref{c4.p.sizen} and Remark \ref{c4.r.2n1}}.
\end{eqnarray*}
Thus, the sizes of syntactic monoids of the submonoids accepted by CSFA with two bpis over a binary alphabet is bounded by $2n(n+1)$, where $n$ is the state complexity of the CSFA. For the class of automata displayed in Subsection 4.3.4, the syntactic monoid size is exactly $2n(n+1)$. Hence, the syntactic complexity of CSFA with two bpis over a binary alphabet is $2n(n+1)$.

As shown in the following example, the syntactic complexity of CSFA with two bpis over an alphabet of size more than two is not $2n(n + 1)$.

\begin{example}\label{c4.e.2and3}
Consider the CSFA $\au$ over the ternary alphabet $\{a, b, c\}$ given in \textsc{Figure} \ref{c4.fig1}. Here, $BPI(\au) = \{q_0, q_3\}$.
One can compute that the syntactic complexity of $\au$ is 110.
\begin{figure}[htb]
\entrymodifiers={++[o][F-]} \SelectTips{cm}{}
\[\xymatrix{
*\txt{} & *++[o][F=]{q_0} \ar[dr]^a \ar[ddl]^{b, \ c} & *\txt{}\\
q_4 \ar[ur]^{a, \ b, \ c} & *\txt{} & q_1\ar[d]^a \ar[dll]_b \ar@/_1.5pc/[ul]_c\\
q_3 \ar[u]^a \ar@/^4pc/[uur]^{b, \ c} & *\txt{} & q_2 \ar[ll]^{a, \ b, \ c}}\]
\caption{A CSFA over ternary alphabet}
\label{c4.fig1}
\end{figure}
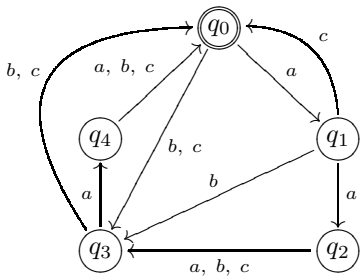
\end{example}

\section{Conclusion}

This work investigates the syntactic complexity of various classes of the submonoids accepted by CSFA, classified by their number of bpis. In fact, we showed that the syntactic complexity of CSFA with at most one bpi is linear. Further, we proved that the syntactic complexity of CSFA with two bpis over a binary alphabet is $2n(n+1)$. In that connection, we obtained a representation for the functions of rank two in the monoid of CSFA with two bpis over a binary alphabet. However, there is a lot more to investigate the syntactic complexity concerning the finitely generated submonoids of a free monoid. For instance, one can target to address the syntactic complexity of CSFA with two bpis over an arbitrary alphabet. In general, one can study the syntactic complexity of CSFA and SFA with more than two bpis.

\end{document}